\documentclass[11pt]{article}
\usepackage[utf8]{inputenc}
\usepackage[margin=1in]{geometry}
\usepackage{mathtools, amsmath, amsthm, graphicx, enumitem, amsfonts, amssymb, verbatim, bbm, braket, xcolor, float, booktabs, multirow, colortbl, complexity, comment, tikz}
\usepackage[linktocpage=true]{hyperref}
\usepackage[backend=biber, style=alphabetic, maxbibnames=99]{biblatex}
\usepackage[capitalise]{cleveref}
\allowdisplaybreaks

\newcommand{\defeq}{\stackrel{\mathrm{\scriptscriptstyle def}}{=}}
\newcommand{\of}[1]{\left( #1 \right)}
\newcommand{\ofc}[1]{\left\{ #1 \right\}}
\newcommand{\ofb}[1]{{\left[#1\right]}}

\newcommand{\pmatch}{\mathrm{PMATCH}}
\newcommand{\match}{\mathrm{MATCH}}
\newcommand{\prd}{\mathrm{PROD}}
\newcommand{\tr}{\mathrm{Tr}}
\newcommand{\bloch}[1]{\vec{B}(#1)}
\newcommand{\maxFM}{\textnormal{FM}}
\newcommand{\maxM}{\textnormal{M}}

\newtheorem{lemma}{Lemma}
\newtheorem{theorem}{Theorem}

\newtheorem{algorithm}{Algorithm}

\theoremstyle{definition}
\newtheorem{remark}{Remark}

\hypersetup{
    colorlinks=true,
    linkcolor=blue,
    filecolor=magenta,      
    urlcolor=cyan,
    citecolor=violet,
}

\begin{document}
\title{Improved Algorithms for Quantum MaxCut via Partially Entangled Matchings}
\author{
  Anuj Apte\thanks{University of Chicago}
  \and
  Eunou Lee\thanks{Korea Institute for Advanced Study}
  \and
  Kunal Marwaha\footnotemark[1]
  \and
  Ojas Parekh\thanks{Sandia National Laboratories}
  \and
  James Sud\footnotemark[1]\;\,\thanks{jsud@uchicago.edu}
}

\date{}
\maketitle

\begin{abstract}
    We introduce a $0.611$-approximation algorithm for Quantum MaxCut and a $\frac{1+\sqrt{5}}{4} \approx 0.809$-approximation algorithm for the EPR Hamiltonian of \cite{king2023}. A novel ingredient in both of these algorithms is to \emph{partially} entangle pairs of qubits associated to edges in a matching, while preserving the direction of their single-qubit Bloch vectors. This allows us to interpolate between product states and matching-based states with a tunable parameter.
\end{abstract}

\section{Introduction}\label{sec:introduction}
Given a graph $G(V, E, w)$ with positive edge weights $w: E \rightarrow \mathbb{R}_+$ and a $2$-local Hamiltonian term $h$, define the $n$-qubit Hamiltonian 
\begin{align*}
    H_G \defeq \sum_{(i,j) \in E(G)} w_{ij} \cdot h_{ij}\,,
\end{align*}
where $h_{ij}$ is the local term $h$ applied on qubits $(i,j)$. For two particular local terms $h$, we are interested in the problem of computing the the maximum energy of $H_G$, which we denote as $\lambda_{max}(H_G)$, for any $G$. This is not an easy task in general: deciding if $\lambda_{max}(H_G)$ is above some threshold with inverse polynomial accuracy is known to be $\QMA$-hard~\cite{piddock2015}. 

In the first problem, we choose the local term
\begin{align*}
    h^{QMC}_{ij} \defeq \frac{1}{2} \left( I_i I_j - X_iX_j - Y_iY_j - Z_iZ_j \right) = 2 \ket{\psi_-}_{ij}\bra{\psi_-}_{ij}\,,
\end{align*}
where $\ket{\psi_-} = \frac{1}{\sqrt{2}} \left(\ket{01} - \ket{10}\right)$ is the singlet state. This problem has recently been studied under the name \emph{Quantum MaxCut} (QMC) \cite{gharibian2019}. In the statistical mechanics literature, Hamiltonians $H_G$ defined by $h^{QMC}$ are instances of the zero-field quantum Heisenberg XXX\textsubscript{1/2} model. The decision version of QMC is $\QMA$-hard~\cite{piddock2015}.

In the second problem,  we choose the local term
\begin{align*}
        h^{EPR}_{ij} \defeq \frac{1}{2} \left( I_i I_j + X_iX_j - Y_iY_j + Z_iZ_j\right)= 2 \ket{\phi_+}_{ij}\bra{\phi_+}_{ij}\,,
\end{align*}
where $\ket{\phi_+} = \frac{1}{\sqrt{2}} \left(\ket{00} + \ket{11}\right)$. This problem, named \emph{EPR} by \cite{king2023}, is thought to be easier than QMC, since Hamiltonians $H_G$ defined by $h^{EPR}$ are \emph{stoquastic} (sign-problem free) \cite{piddock2015, cubitt2016}. In fact, it is not yet clear if EPR can be solved in polynomial time. 

In lieu of exactly computing $\lambda_{max}(H_G)$, we may try to \emph{approximate} this value. In both problems, the local term $h$ is positive semidefinite, so $\lambda_{max}(H_G) \ge 0$. We judge an approximation by its \emph{approximation ratio}. Suppose we can find efficiently computable functions $\ell, u$ such that for all graphs $G$,
\begin{align*}
    0 \le \ell(G) \le \lambda_{max}(H_G) \le  u(G)\,.
\end{align*}
Then, the approximation ratio $\alpha$ is at least
\begin{align*}
    \alpha \ge \min_G \frac{\ell(G)}{\lambda_{max}(H_G)} \ge  \min_G \frac{\ell(G)}{u(G)}\,.
\end{align*}
In most works, the upper bound $u$ is determined by solving a semidefinite programming (SDP) relaxation of the maximization problem \cite{brandao2014, gharibian2019, parekh2021, takahashi2023, watts2024,huber2024}.
The source of such SDP relaxations are generally hierarchies of SDPs that provide increasingly better upper bounds $u$ at the cost of solving larger-sized SDPs. The quantum moment-SOS hierarchy, typically based on Pauli operators, is widely employed for quantum local Hamiltonian problems. This hierarchy is an instance of the NPA hierarchy~\cite{navascues2008} and is also known as the quantum Lasserre hierarchy~\cite{parekh2021}. The second level of the quantum moment-SOS hierarchy is necessary for good approximations, since critical monogamy-of-entanglement properties begin to emerge at that level~\cite{parekh2021, parekh2022}.

The lower bound $\ell$ is usually determined by an algorithm that prepares a state with this energy. Several approximation algorithms have been proposed for QMC and EPR, using product states~\cite{gharibian2019, parekh2021, parekh2022, king2023}, matchings~\cite{anshu2020, parekh2021, parekh2022, lee2024, jorquera2024, gribling2025}, and short variational circuits~\cite{anshu2020,lee2022,king2023,ju2025,gribling2025}.
During the preparation of this manuscript, the best known approximation ratios for these problems were improved to $\alpha \ge 0.603$ for QMC~\cite{gribling2025} and $\alpha \ge \frac{1+\sqrt{5}}{4} \approx 0.809$ for EPR~\cite{ju2025}. 

We provide a new algorithm for both problems. For each algorithm, we start from a good product state. We then choose a matching in $G$ and \emph{partially} rotate matched qubits towards an entangled state. Our novel contribution is to perform this rotation such that individual single-qubit Bloch vectors are preserved, up to a rescaling of magnitude. This allows us to evaluate the energy of our entangled state in terms of the energy of the product state. For EPR, we use a \emph{fractional} matching and provide a circuit-based algorithm. For QMC, we explicitly describe a tensor-product of single and two-qubit states, which interpolates between the product state and an \emph{integer} matching-based state with a single tunable parameter. For QMC, we additionally introduce the technique of choosing a maximum-weight matching with respect to \emph{rescaled} edge weights. One interpretation of this approach is that we widen the search space of our algorithm to include matchings other than the maximum-weight matching of $G$. Combining these techniques allows us to achieve state-of-the-art approximation ratios on each problem:

\begin{theorem}
\label{thm:epr_improved}
    For EPR, $\alpha \ge \frac{\varphi}{2} \approx 0.809$,
    where $\varphi \defeq \frac{1 + \sqrt{5}}{2} \approx 1.618$ is the golden ratio.
\end{theorem}

\begin{theorem}
\label{thm:qmc_improved}
    For QMC, $\alpha \ge 0.611$.
\end{theorem}

\begin{remark}
\cref{thm:epr_improved} was shown simultaneously and independently by \cite{ju2025} using a different algorithm. The algorithm of \cite{ju2025} also \emph{partially} entangles edges while preserving single-qubit marginals. However, we believe both our algorithm and analysis are simpler. 
\end{remark}

\begin{remark}
    \cref{thm:qmc_improved} uses an improved SDP upper bound and the 13th level of the Lasserre hierarchy from \cite{gribling2025}.
    After release of this manuscript, Bakshi et al \cite{bakshi2026,apte2025b} proved an upper bound on all QMC Hamiltonians. \cite{apte2025b} shows how these bounds can be efficiently incorporated into SDPs via a separation oracle. With respect to this method applied to the second-level Lasserre, our QMC algorithm achieves $\alpha \ge 0.614$ (\cref{sec:analysis/qmc}).
\end{remark}

Our approximation algorithm finds a quantum state $\rho_G$ such that $\tr\of{H_G \rho_G} \ge \ell(G)$. For the EPR Hamiltonian, the analysis of our algorithm is optimal: there exist graphs $G$ (such as the single-edge graph) where
\begin{align*}
    \frac{\ell(G)}{u(G)} \le  \frac{\bra{\psi_G} H_G \ket{\psi_G}}{\lambda_{max}(G)} = \frac{1 + \sqrt{5}}{4}\,.
\end{align*}
For QMC, however, it is possible that a better analysis of this algorithm (particularly of $u(G)$) would give a larger approximation ratio. 

In the remainder of this work, we introduce necessary notation, describe each algorithm, and prove \Cref{thm:epr_improved,thm:qmc_improved}. 
We defer some technical lemmas to the appendix.

\section{Preliminaries}\label{sec:preliminaries}

\subsection{Graph theory}\label{sec:preliminaries/graph_theory}
Let $G(V,E,w)$ denote a graph with vertex set $V$, edge set $E$, and positive edge weights $w: E \rightarrow \mathbb{R}_+$. In this work, we always take $V=\ofb{n} \defeq \{1,2,\dots,n\}$. We let $W_G \defeq \sum_{(i,j) \in E} w_{ij}$. When $G$ is inferred by context, we simply write $W$. Define $N(v)$ as the set of neighbors of a vertex $v \in V$. For convenience, we index edges by either $e$ or $(i,j)$, or simply $ij$ in subscripts. 

A fractional \emph{matching} of a graph $G(V,E,w)$ is a function $m: E(G) \to [0,1]$ that assigns a value to each edge in $G$, such that the sum over values $\sum_{j\in N(i)} m_{ij}$ is at most $1$ for any vertex $i$. If $m_{ij} \in \{0,1\}$ for all edges $(i,j)$, we say that $m$ is an \emph{integral} matching. For an integral matching $m$, we say an edge $(i,j)$ is in the matching if $m_{ij}=1$.
The \emph{weight} $\textrm{Wt}(m)$ of a matching $m$ is defined as $\textrm{Wt}(m) \defeq \sum_{(i,j) \in E(G)} w_{ij} \, m_{ij}$.   We define $\maxFM_G$ (or $\maxM_G$) to be the maximum total weight of any fractional (or integral) matching of $G$. For convenience, we sometimes let $\maxM_G$ also denote the set of edges in the maximum-weight integral matching of $G$. 
When the graph is clear from context we drop the subscript $G$. Optimal integral and fractional matchings can be computed in polynomial time, for example with linear programming~\cite{edmonds1965}.

\subsection{Quantum computation}\label{sec:preliminaries/quantum_computation}
We refer to $\vec{\sigma}=\of{X,Y,Z}$ as the canonical Pauli matrices. As such, we can define the \emph{Bloch vector} $\bloch{\rho}$ of any $1$-qubit density matrix $\rho$ as the vector $(b_x, b_y, b_z)$ in the unit sphere $\textrm{S}^2$ such that 
    \begin{align*}
        \rho = \frac{1}{2} (I + \bloch{\rho} \cdot \vec{\sigma}) \defeq \frac{1}{2} (I + b_x X + b_y Y + b_z Z)\,.
    \end{align*}
Pure states correspond to unit Bloch vectors. 
Let $\ket{\psi_-} = \frac{1}{\sqrt{2}} \left(\ket{01} - \ket{10}\right)$ be the \emph{singlet state},  and  $\ket{\phi_+} = \frac{1}{\sqrt{2}} \left(\ket{00} + \ket{11}\right)$ be the \emph{EPR state}.

\subsection{Previous algorithms}\label{sec:preliminaries/previous_algorithms}
In our work, we use two previously known algorithms as subroutines:
\begin{enumerate}
    \item  The algorithm $\prd$ takes a graph $G(V,E,w)$ as an input, and outputs a good product state $\rho_{PROD}$. For EPR, we define $\rho_{PROD} \defeq \ket{0}^{\otimes n} \bra{0}^{\otimes n}$.
    For QMC, we define $\prd$ to be the output of the \cite{gharibian2019} rounding algorithm applied to the level-$k$ quantum moment-SOS hierarchy, as in~\cite{parekh2022} and~\cite{lee2024}.
    This is the only part of our algorithm for QMC that relies on an SDP. As described in~\cite[Sec 2.2]{parekh2022}, the output from the SDP can be interpreted as a \emph{pseudo-density matrix}. It is called a \emph{pseudo}-density matrix because it satisfies only some of the constraints of a valid density matrix.
    \item The algorithm $\match$ takes a graph $G(V,E,w)$ as an input, and outputs a product of 2-qubit states $\rho_{MATCH}$. It was first formally proposed in~\cite{lee2024}. The algorithm first finds a maximum-weight integral matching $\maxM_G$. For EPR (and for QMC), it outputs the tensor product of $\ket{\phi_+}_{ij}\bra{\phi_+}_{ij}$ (and $\ket{\psi_-}_{ij}\bra{\psi_-}_{ij}$ for QMC) for every pair of vertices $(i,j)$ in the matching, and the maximally mixed state for every vertex $i$ not in the matching. This obtains energy $2$ on matched edges and $1/2$ on unmatched edges, for a total of $(3\maxM_G+W_G)/2$.
\end{enumerate}

\section{Our Algorithms}\label{sec:algorithms}

To motivate our algorithms, consider the following simple graph $G$: 
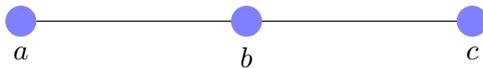
\vspace{2em}
\begin{center}
\begin{tikzpicture}[scale=1.5]
  \node[draw=blue!50, fill=blue!50, circle, inner sep=4pt, label=below:$a$] (a) at (0,0) {};
  \node[draw=blue!50, fill=blue!50, circle, inner sep=4pt, label=below:$b$] (b) at (2,0) {};
  \node[draw=blue!50, fill=blue!50, circle, inner sep=4pt, label=below:$c$] (c) at (4,0) {};
  \draw[black] (a) -- (b);
  \draw[black] (b) -- (c);
\end{tikzpicture}
\end{center}
\vspace{.5em}
This graph is bipartite, so QMC and EPR are equivalent under local rotations~\cite{king2023}. We thus focus on EPR for simplicity. It is easy to compute that $\lambda_{max}(H_G)=3$. The optimal product state $\rho_{PROD}$ is $\rho_0^{\otimes 3}$, where $\rho_0 \defeq \ket{0}\bra{0}$, and achieves energy $2$. The algorithm $\match$ computes the maximum matching $\maxM_G = \ofc{(a,b)}$ and returns the state $\rho_{MATCH}$ as described in \cref{sec:preliminaries}. This state gains energy $2$ on edge $(a,b)$ and $1/2$ on edge $(b,c)$, achieving total energy $5/2$. Thus, the algorithm of~\cite{lee2024}, which returns the better of $\rho_{PROD}$ and $\rho_{MATCH}$, achieves $5/6\approx 0.833$ of the optimal energy. 

Our algorithms do better by interpolating between $\rho_{PROD}$ and $\rho_{MATCH}$, rather than taking the better of the two. For our example $G$, the unitary
\begin{align*}
    U = e^{i \theta \of{\frac{X_a-Y_a}{\sqrt{2}}\otimes \frac{X_b-Y_b}{\sqrt{2}}}}\,,
\end{align*}
takes $\rho_{PROD}$ to $\rho'$, a tensor product of a two-qubit state $\rho_{ab}$ and a single-qubit state $\rho_{c}$. The parameter $\theta$ sets the entanglement for $\rho_{ab}$: when $\theta=0$, $\rho_{ab}$ is a product state; when $\theta=\pi/4$, $\rho_{ab}$ is fully rotated into the EPR state. As such, we view this approach as smoothly interpolating between $\rho_{PROD}$ and $\rho_{MATCH}$ in superposition. It can be easily verified that the energy obtained by $\rho'$ on edge $(a,b)$ is $\tr\ofb{h^{EPR}_{ab}\rho'} = 1+2\cos\theta\sin\theta$, and the single qubit marginals of $a$ and $b$ are given by 
\begin{align*}
    \rho'_a \defeq \tr_{bc}\ofb{\rho'} = \cos{2\theta}\rho_0 + \sin^2{\theta}I = \rho'_b\,.
\end{align*}
In particular, note that the single qubit marginals are simply rescaled and shifted by the identity. This fact aids in the analysis:  we can still evaluate the energy of $(b,c)$ in terms of $\rho_{PROD}$, even though the edge is not in our matching:

\begin{align*}
    \tr\ofb{h^{EPR}_{bc} \rho'}
    &= \tr\ofb{h^{EPR}_{bc} \of{\rho'_b \otimes \rho'_c}}\\
    &= \tr\ofb{h^{EPR}_{bc} \of{\rho'_b \otimes \rho_0}} \\
    &= \tr\ofb{h^{EPR}_{bc} \of{
    \of{\cos{2\theta}\rho_0
    + \sin^2{\theta}\,I} \otimes \rho_0}} \\
    &= \tr\ofb{h^{EPR}_{bc} \of{
    \cos{2\theta}\rho_0\otimes\rho_0
    + \sin^2{\theta}\,I\otimes\rho_0}} \\
    &= \cos{2\theta}+\sin^2\theta
    = \cos^2\theta\,.
\end{align*}

Thus the total energy is given by 
\begin{align*}
\tr\ofb{(h^{EPR}_{ab}+h^{EPR}_{bc}) \rho'}=1+2\cos\theta\sin\theta+\cos^2\theta\,.
\end{align*}
The parameter $\theta$ can be optimized over; in this example, taking $\theta\approx .554$ yields
\begin{align*}
    \tr\ofb{(h^{EPR}_{ab}+h^{EPR}_{bc}) \rho'} \approx 2.618,
\end{align*}
outperforming both $\rho_{PROD}$ and $\rho_{MATCH}$.

In this example, the analysis simplifies nicely because the initial product state is the symmetric state $\rho_0^{\otimes n}$, and because the graph is bipartite. For EPR on general graphs, we provide a circuit-based algorithm that prepares an entangled state based on a \emph{fractional} matching in the graph. This state still works by \emph{partially} entangling edges according to the matching while preserving single-qubit marginals, but it does not have a simple interpretation in terms of $\rho_{MATCH}$. We leave open the possibility that a $(\frac{1+\sqrt{5}}{4})$-approximation algorithm can be achieved with a tensor product of one and two-qubit states that simply interpolates between $\rho_{PROD}$ and $\rho_{MATCH}$. For QMC, our algorithm does indeed interpolate between $\rho_{PROD}$ and $\rho_{MATCH}$. However the algorithm and analysis become slightly more complicated. Crucial to our analysis is the following lemma:

\begin{lemma}[Energy obtained by QMC algorithm on matched and unmatched edges]\label{lem:qmc_energy_by_edge}
Given two $1$-qubit pure states $\rho_i, \rho_j$ and a real parameter $\theta \in [0,\pi/2]$, there exists a $2$-qubit pure state $\rho_{ij}$ such that 
\begin{align*}
            \tr\ofb{{h^{QMC} \rho_{ij}}} &= \frac{
            (1 + \sin{\theta}) \left(1 - \bloch{\rho_i} \cdot \bloch{\rho_j}\right)}{2} \,,\\
        \bloch{\tr_j\ofb{\rho_{ij}}} &= \cos{\theta} \cdot \bloch{\rho_i}  \,,\\
        \bloch{\tr_i\ofb{\rho_{ij}}} &= \cos{\theta} \cdot \bloch{\rho_j}  \,.
\end{align*}
\end{lemma}

The last two constraints in \Cref{lem:qmc_energy_by_edge} indicate that the Bloch vectors of the single-qubit marginals of $\rho_{ij}$ are rescaled Bloch vectors of $\rho_i$ and $\rho_j$, respectively. This is analogous to the argument in our example, where the marginals are rescaled and shifted by the identity. This again allows us to compute the energy of unmatched edges in terms of $\rho_{PROD}$. The parameter $\theta\,$ again sets the entanglement in $\rho_{ij}$: when $\theta = 0$, $\rho_{ij}$ is a product state; when $\theta = \pi/2$, $\rho_{ij}$ is maximally entangled. 
The proof of \Cref{lem:qmc_energy_by_edge} is deferred to \Cref{apx:proom_om_lemma_qmc_energy_by_edge}.

\subsection{EPR}\label{sec:algorithms/epr}

Using the intuition from our example, we introduce the following algorithm for EPR:
\begin{algorithm}[Fractional matching algorithm for EPR]\label{alg:matching_algo_epr}
\vspace{.5em}
    Given a weighted graph  $G(V, E, w)$:
    \begin{enumerate}
    \item Define the $1$-qubit Hamiltonian
    \begin{align*}
        P \defeq \frac{1}{\sqrt{2}}(X - Y)\,,
    \end{align*}
    and the angle
    \begin{align}\label{eq:epr_frac_specialtheta}
        \theta \defeq \frac{\ln\of{\varphi}}{2}  \approx 0.240\,,
    \end{align}
    where $\varphi = \frac{1+\sqrt{5}}{2}$ is the golden ratio.
        \item Find a fractional matching $(m_{uv})_{(u,v)\in E}$ of maximum weight (e.g., via linear programming).
        \item Output the state 
    \begin{align}\label{eq:epr_chi}
        \ket{\chi} \defeq \prod_{(u,v) \in E} e^{i \gamma_{uv} P_u P_v} \ket{0}^{\otimes n}\,,
    \end{align}
    where 
\begin{align}\label{eq:epr_frac_match_gamma_assignment}
    \gamma_{uv} &= \frac{1}{2} \cos^{-1} \exp\ofb{-\theta \cdot m_{uv}}  \  \in  \ [0, \pi/2] \,.\end{align} 
    \end{enumerate}
\vspace{.5em}
\end{algorithm}
The unitary $e^{i \gamma P_u P_v}$ rotates the state $\ket{00}_{ij}$ towards the Bell state $\ket{\phi^+}$; this circuit was proposed in the approximation algorithm of \cite{king2023}.

\subsection{QMC}\label{sec:algorithms/qmc}
Our algorithm for QMC relies on the following subroutine, which we call $\pmatch$ since it places partially entangled 2-qubit states on the edges of a matching, while $\match$ places maximally entangled states:
\begin{algorithm}[$\pmatch$]\label{alg:matching_algo_qmc}
Given a weighted graph $G(V, E, w)$ and a real parameter $\theta \in \ofb{0, \pi/2}$: 
\vspace{.5em}
    \begin{enumerate}
        \item Prepare the state $\rho_{PROD} = \bigotimes_{i \in [n]} \rho_i$ by running the level-$k$ $\prd$ algorithm on QMC from \cref{sec:preliminaries}.
        \item For each edge $(i,j) \in E$, compute $t_{ij} \defeq \bloch{\rho_i} \cdot \bloch{\rho_j}$ from $\rho_{PROD}$.
        \item Find a maximum-weight integral matching $\widetilde{M}\subseteq E$ on the \emph{reweighted} graph $G(V, E', \widetilde{w})$, where 
        \begin{align}\label{eq:qmc_algo_rescaled_weights}
            \widetilde{w}_{ij} = w_{ij}  \cdot \left(  \frac{\sin{\theta}\of{1-t_{ij}\of{1+\sin{\theta}}}}{2} \right)^+\,,
        \end{align}
        and $E' = \{(i,j) \in E\ |\ \widetilde{w}_{ij} > 0\}$. Here, we use the notation $(\cdot)^+ \defeq \max(0, \cdot)$ from \cite{lee2024}.
        \item Let $\rho_{PMATCH}\of{\theta}$ be the state starting from $\rho_{PROD}$, but for each edge $(i,j)$ in $\widetilde{M}$, replace $\rho_i \otimes \rho_j$ with the state $\rho_{ij}$ described in \cref{lem:qmc_energy_by_edge}, parametrized by $\rho_i, \rho_j, \theta$. Additionally, replace each unmatched qubit $\rho_i$ with $\cos\theta \rho_i + \frac{1-\cos\theta}{2} I $.
        \item Output the state $\rho_{PMATCH}\of{\theta}$.
    \end{enumerate}
\vspace{.5em}
\end{algorithm}

Our algorithm for QMC chooses the better of the algorithms $\pmatch$ and $\match$:

\begin{algorithm}[Combined algorithm for QMC]\label{alg:algo_qmc}
\vspace{.5em}
    Given a weighted graph $G(V, E, w)$:
    \begin{enumerate}
        \item Prepare the state $\rho_{MATCH}$ by running the $\match$ algorithm from \cref{sec:preliminaries}.
        \item Prepare the state $\rho_{PMATCH}\of{\theta}$ by running $\pmatch$ with parameter $\theta=0.2847$.
        \item Output the state out of $\rho_{MATCH}$ and $\rho_{PMATCH}(\theta)$ obtaining larger energy on the QMC Hamiltonian $H_G$.
    \end{enumerate}
\vspace{.5em}
\end{algorithm}

\section{Analysis}\label{sec:analysis}

We upper-bound the maximum energy $\lambda_{max}(H_G)$ of QMC and EPR Hamiltonians using a quantifiable \emph{monogamy of entanglement}:
\begin{lemma}[{Monogamy of Entanglement, e.g.~\cite[Lemma 1]{lee2024}}]\label{lem:fractional_monogamy}
    For both EPR and QMC, and for all graphs $G(V,E,w)$, we have
    \begin{align*}
        \lambda_{max}(H_G) \le W_G + \maxFM_G\,.
    \end{align*}
\end{lemma}

For QMC, our analysis improves with better upper bounds on $\lambda_{max}(H_G)$. We use an inequality of the form $W_G + \maxM_G/d$;  this first appears in \cite[Lemma 4]{lee2024} with the constant $d=4/5$. By a detailed analysis and numerical verification on graphs with up to $13$ vertices, this constant was recently improved to $14/15$ in \cite{gribling2025}:
\begin{lemma}[{Strengthened monogamy of entanglement, \cite[Lemma 3.10]{gribling2025}}]\label{lem:strengthened_monogamy} 
    For QMC, and for all graphs $G(V,E,w)$, we have
    \begin{align*}
         \lambda_{max}(H_G) \le W_G + \frac{\maxM_G}{d}\,,
    \end{align*}
    where $d = \frac{14}{15}$.
\end{lemma}
\begin{remark}\label{rem:d_sdp}
    The bound in \cref{lem:strengthened_monogamy} holds with $d=\frac{10}{11}$ when the left-hand side is replaced by the SDP energy from the $2$nd level quantum moment-SoS hierarchy and with $d=\frac{14}{15}$ for the $13$th level \cite[Theorem 3.11]{gribling2025}.
\end{remark}
\begin{remark}\label{rem:d_equals_1}
    After the release of the original manuscript, the constant $d$ in \cref{lem:strengthened_monogamy} was improved to $1$ by \cite{apte2025b,bakshi2026}. The $2$nd level quantum moment-SoS hierarchy can be augmented with a separation oracle, such that is also obeys \cref{lem:strengthened_monogamy} with $d=1$ \cite{apte2025b}.
\end{remark}

\subsection{EPR}\label{sec:analysis/epr}
We now prove \cref{thm:epr_improved}. Given the state in \cref{eq:epr_chi}
\begin{align}\label{eqn:king_form}
   \ket{\chi} \defeq \prod_{(i,j) \in E} e^{i \gamma_{ij} P_i P_j} \ket{0}^{\otimes n}\,,
\end{align}
\cite[Lemma 9]{king2023} showed that the energy on the local term $h^{EPR}_{ij}$ is at least
\begin{align}\label{eq:epr_single_edge_energy}
    \bra{\chi} h_{ij}^{EPR} \ket{\chi} \ge \frac{1}{2}\of{1 + A_{ij} B_{ij} + \sin 2\gamma_{ij} \cdot (A_{ij} + B_{ij})}\,,
\end{align}
where
\begin{align*}
     A_{ij} &\defeq  \prod_{k \in N(i) \setminus \{j\}} \cos 2\gamma_{ik}\,, \\
     B_{ij} &\defeq  \prod_{k \in N(j) \setminus \{i\}} \cos 2\gamma_{jk}\,.
\end{align*}
    
In \Cref{alg:matching_algo_epr}, the output state $\ket{\chi}$ is in the form of \cref{eqn:king_form}. Using the angles $\gamma_{ij}$ specified by \cref{eq:epr_frac_match_gamma_assignment}, we have
\begin{align}\label{eq:epr_a_b_ineq}
        A_{ij} &= \exp\ofb{-\theta \!\!\sum_{k \in N(i) \setminus \{j\}} \!\! m_{ik}} \ge \exp\ofb{-\theta (1 - m_{ij})}\,,\\
         B_{ij} &= \exp\ofb{-\theta \!\!\sum_{k \in N(j) \setminus \{i\}} \!\! m_{jk}} \ge \exp\ofb{-\theta (1 - m_{ij})}\,,\\
         \sin 2 \gamma_{ij} &= \sqrt{1 - \cos^2 2 \gamma_{ij}} = \sqrt{1 - \exp\ofb{-2 \theta m_{ij}}}\,.
\end{align}
The inequalities in the first two lines follow because $m$ is a matching. For example, we have $m_{ij} + \sum_{k \in N(i) \setminus \{j\}} m_{ik} \le 1$. Using \cref{eq:epr_single_edge_energy}, the energy of $\ket{\chi}$ on $h_{ij}^{EPR}$ is at least
\begin{align}\label{eq:epr_fm_algo_energy}
    T(\theta, m_{ij}) \defeq \frac{1}{2}\of{1 + \exp\ofb{-2\theta(1 - m_{ij})} + 2 \sqrt{1 - \exp\ofb{-2 \theta m_{ij}}} \exp\ofb{-\theta (1 - m_{ij})}}\,.
\end{align}
We combine \cref{eq:epr_fm_algo_energy} with \cref{lem:fractional_monogamy} to bound the approximation ratio on any graph $G(V,E,w)$:
\begin{align*}
    \frac{\bra{\chi}H_G\ket{\chi}}{\lambda_{max}(H_G)} \ge \frac{\sum_{(i,j) \in E} w_{ij} \cdot T(\theta, m_{ij})
    }{\sum_{(i,j)\in E} w_{ij} \of{1+m_{ij}}}\,.
\end{align*}
Each term in the numerator and denominator is positive, so the approximation ratio is at least the approximation ratio of the worst edge
\begin{align}\label{eq:epr_analysis_theta}
   \min_{(i,j) \in E} \frac{T(\theta, m_{ij})}{1 + m_{ij}} \ge \min_{x \in [0,1]} \frac{T(\theta, x)}{1 + x}\,.
\end{align}
Recall from \cref{eq:epr_frac_specialtheta} that we take $\theta=\frac{1}{2}\ln\varphi$. We empirically identified $\theta$ as the angle that maximizes the approximation ratio in this analysis. Substituting this value into $T(\theta, x)$ in \cref{eq:epr_fm_algo_energy}, the RHS of \cref{eq:epr_analysis_theta} yields the minimization problem
\begin{align*}
    \min_{x \in [0,1]} \ofb{
    \frac{1}{2(1+x)} \left(
    1 + \varphi^{-(1-x)} + 2\sqrt{1 - \varphi^{-x}}
    \cdot \varphi^{-\frac{1}{2}(1-x)}
    \right)},
\end{align*}
whose value is found to be $\varphi/2$ by the following lemma:

\begin{lemma}\label{lemma:epr_numerical_inequality}
Given $\varphi \defeq \frac{1+\sqrt{5}}{2}$, we have that
\begin{align}\label{eq:epr_min_problem}
    \min_{x \in [0,1]} \ofb{
    \frac{1}{2(1+x)} \left(
    1 + \varphi^{-(1-x)} + 2\sqrt{1 - \varphi^{-x}}
    \cdot \varphi^{-\frac{1}{2}(1-x)}
    \right)
    } =\frac{1+\sqrt{5}}{4} = \frac{\varphi}{2}\,.
\end{align}  
\end{lemma}
The proof of \Cref{lemma:epr_numerical_inequality} is deferred to \cref{apx:proom_om_lemma_epr_numerical}.

\subsection{QMC}\label{sec:analysis/qmc}
We now analyze \cref{alg:algo_qmc} and prove  \cref{thm:qmc_improved}. We first introduce some helpful notation. 
Fix a graph $G(V,E,w)$.
Let $\tilde{\rho}$ be the pseudo-density matrix outputted by solving the SDP of the level-$k$ quantum moment-SOS hierarchy (described in \cref{sec:preliminaries/previous_algorithms}).
We define 
\begin{align*}
g_{ij} \defeq \tr[h_{ij}^{QMC} \tilde{\rho}]\,, \quad\quad s_{ij}\defeq \frac{1}{3}\tr\ofb{(X_iX_j+Y_iY_j+Z_iZ_j)\tilde{\rho}}\,.
\end{align*}
The first value is the energy obtained by $\tilde{\rho}$ on $h_{ij}^{QMC}$; the second value is the expected value of $\tilde{\rho}$ with respect to the traceless components of the term $h_{ij}^{QMC}$. It is straightforward to show that $g_{ij} = \frac{1}{2} (1 - 3 s_{ij})$. 
The output of the SDP gives an upper bound to the optimal energy:
\begin{align}
\lambda_{max}(H_G) \le \sum_{(i,j)\in E} w_{ij} \cdot g_{ij} = \sum_{(i,j)\in E} w_{ij} \cdot \frac{1-3s_{ij}}{2}\,.\label{eqn:sdpupperbound}
\end{align}
Suppose \cref{lem:strengthened_monogamy}, with the left-hand side replaced by the SDP energy, is true for some $d$. Then, consider a \emph{reweighting} $\widetilde{G}(V,E,u)$ of the graph $G$ where $u_{ij} \ge 0$. Let $x \in \mathbb{R}^E$ be a vector with 
\begin{align*}
    x_{ij}  \defeq d \of{g_{ij}-1}^+ = d \of{-\frac{1+3s_{ij}}{2}}^+.
\end{align*}
Here, we again use the notation $(\cdot)^+ \defeq \max(\,\cdot\,,0)$ from~\cite{lee2024}, which ensures 
\begin{align*}
\of{-\frac{1+3s_{ij}}{2}}^+  \in \ofb{0,1}\,,  \quad\quad \forall (i,j)\in E.
\end{align*}

It is shown (e.g. in \cite[Lemma 3.4]{gribling2025}) that $x$ is in the \emph{integral} matching polytope of $\widetilde{G}$. This means that $x$ is a convex combination of integral matchings $\ofc{m_\ell}_{1 \le \ell \le k}$. Since the weight of any integral matching $m_{\ell}$ on $\widetilde{G}$ is at most $\maxM_{\widetilde{G}}$, the weight of $x$ on $\widetilde{G}$ is also at most  $\maxM_{\widetilde{G}}$. Thus,
\begin{align}\label{eq:s_matching_relation}
    \sum_{(i,j)\in E} u_{ij}\cdot d \cdot \of{-\frac{1+3s_{ij}}{2}}^+ \leq \maxM_{\widetilde{G}}\,.
\end{align}
\cref{eq:s_matching_relation} in this context was first used in~\cite{parekh2021} and also appears in~\cite[Lemma 4]{lee2024} and \cite[Lemma C.2]{jorquera2024}.

We now find a lower bound for the energy obtained by the matching state $\rho_{PMATCH}\of{\theta}$ in $\pmatch$.  We consider the energy on a case-by-case basis in the following lemma:
\begin{lemma}\label{lem:qmc_energy_by_edge_all}
Consider the state $\rho_{PMATCH}$ and matching $\widetilde{M}$ from \Cref{alg:matching_algo_qmc}.
Then the energy of $\rho_{PMATCH}$ on an edge $(i,j)$ is exactly
\begin{align*}
    \tr\ofb{h_{ij}^{QMC} \rho_{PMATCH}} =\begin{cases}
		\frac{1}{2}\of{1+\sin\theta}\of{1-t_{ij}}, & (i,j) \in  \widetilde{M}\,,\\
            \frac{1}{2}( 1- t_{ij} \cos^2\theta ), &(i,j) \notin \widetilde{M}
		 \end{cases}
\end{align*}
    where, $t_{ij}$ is defined in \cref{alg:matching_algo_qmc}.
\end{lemma}
\begin{proof}
The case when $(i,j) \in \widetilde{M}$ follows directly from how we choose $\rho_{PMATCH}$ in Algorithm \ref{alg:matching_algo_qmc}, Step 4. For all other edges, recall that $t_{ij}$ is the inner product of the Bloch vectors of qubits $i$ and $j$ with respect to the product state $\rho_{PROD}$. The reduced density matrix of $\rho_{\pmatch}$ on qubits $i$ and $j$ is a product state $\sigma_i \otimes \sigma_j$.  By \cite[Lemma 10]{king2023}, the energy of $\sigma_i \otimes \sigma_j$ on $h^{QMC}_{ij}$ is  
    \begin{align*}
        \frac{1}{2}\of{1-\bloch{\sigma_i} \cdot \bloch{\sigma_j}} \,.
    \end{align*}
By construction, both Bloch vectors are shrunk by a factor of $\cos\theta$.
\end{proof}

\cref{lem:qmc_energy_by_edge_all} allows us to compute the energy obtained by the state $\rho_{PMATCH}\of{\theta}$:
\begin{align}
    \pmatch(\theta) &\defeq  \sum_{(i,j) \in E} w_{ij} \cdot \tr[h_{ij}^{QMC} \rho_{PMATCH}] \nonumber
     \\
     &=     \sum_{(i,j)\in \widetilde{M}}\frac{w_{ij}}{2}\of{1+\sin\theta}\of{1-t_{ij}} 
     + \sum_{(i,j)\notin \widetilde{M}} \frac{w_{ij}}{2}\of{1-t_{ij}\cos^2\theta } 
     \nonumber,\\
     &= \sum_{(i,j)\in E}\frac{w_{ij}}{2}\of{1- t_{ij}\cos^2\theta} 
     +
     \sum_{(i,j)\in \widetilde{M}}\frac{w_{ij}}{2}\left( \of{1+\sin\theta}\of{1-t_{ij}} - (1 - t_{ij} \cos^2 \theta)\right) \nonumber,
     \\
     &= \sum_{(i,j)\in E}\frac{w_{ij}}{2}\of{1- t_{ij}\cos^2\theta} 
     +
     \sum_{(i,j)\in \widetilde{M}}\frac{w_{ij}}{2}\of{\sin{\theta}\of{1-t_{ij}\of{1+\sin{\theta}}}}  
     \label{eq:algp_pre_matching_sub}\,. 
\end{align}
Note that $\widetilde{M}$ is a maximum-weight \emph{integral} matching with respect to the rescaled weights defined in~\cref{eq:qmc_algo_rescaled_weights}. Thus, we can invoke \cref{eq:s_matching_relation} to replace the sum over $\widetilde{M}$ in \cref{eq:algp_pre_matching_sub} with a sum over $E$. Since the rescaled weights are positive for all edges in $\widetilde{M}$, we have
\begin{align}\label{eq:lower_bound_algp_1}
    \pmatch(\theta) &\ge \sum_{(i,j) \in E} \frac{w_{ij}}{2}
    \Bigg(1-  t_{ij}\cos^2\theta + d \cdot \big(\sin{\theta}\of{1-t_{ij}\of{1+\sin{\theta}}}\big)   \of{-\frac{1+3s_{ij}}{2}}^+\Bigg).
\end{align}
Note that $t_{ij}$ is a random variable because $\prd$ is a randomized algorithm. The analysis of $\prd$ in~\cite{gharibian2019} is based on~\cite{briet2010}, which shows that $\mathbb{E}\ofb{t_{ij}} = F(s_{ij})$. A definition of $F$ in the context of QMC appears in~\cite[Equation 3]{parekh2022}. For our purposes, it is enough to note that $F$ is an odd function because it has the form
\begin{equation}\label{eq:F}
F(s) = c\cdot s \cdot G(s^2)\,,
\end{equation}
where $c > 0$ is a constant and $G \geq 0$ is a hypergeometric function. Since \cref{eq:lower_bound_algp_1} is linear in $t_{ij}$, we use linearity of expectation to conclude

\begin{align}\label{eq:e1_qmc}
    \mathbb{E}\ofb{\pmatch(\theta)} &\ge \sum_{(i,j) \in E} \frac{w_{ij}}{2}
    \Bigg(1- F(s_{ij})\cos^2\theta  + d \sin{\theta}\left(1-F(s_{ij})\of{1+\sin{\theta}}\right)
    \of{-\frac{1+3s_{ij}}{2}}^+\Bigg).
\end{align}
The following lemma lets us analyze the performance of \Cref{alg:algo_qmc}, which takes the better of two algorithms, $\match$ and $\pmatch$. 
\begin{lemma}[Reducing worst-case bounds to a single edge]\label{lem:minimax_reduction}
    Suppose we have a collection of $k$ approximation algorithms
    $\{A_{\ell}\}_{\ell \in \ofb{k}}$ and an upper bound on the maximum
    energy $\lambda_{max}(H_G)\le \sum_{e \in E} w_e b_e$,
    where $b_e \in B = [0, b_{max}]$ for all $e \in E$, and
    $b_{max}$ is a constant. Furthermore, suppose that the expected
    total energy obtained by algorithm $A_{\ell}$ is at least
    $\sum_{e \in E} w_e\, a_{\ell}(b_e)$, where $a_{\ell}$
    is a deterministic function satisfying $a_{\ell}(0)\ge 0$.
    Then, by running each of $A_{\ell\in[k]}$ and taking the output
    with the maximum energy, we can obtain an approximation ratio
    in expectation
    \begin{align*}
        \alpha \ge \max_{\mu_{\ell}} \;\min_{b\in B\setminus\{0\}}
        \sum_{\ell} \mu_{\ell} \frac{a_{\ell}(b)}{b},
    \end{align*}
    where the maximum is taken over valid probability distributions
    \begin{align*}
        \{\ \mu\ |\ \sum_{\ell \in[k]}\mu_{\ell}=1,\;
        0\le \mu_{\ell}\le 1 \; \forall \ell\in[k]\ \}\,.
    \end{align*}
\end{lemma}

\begin{proof}
Given a specific graph $G$, outputting the best of $A_{\ell}$ is at least as good as outputting the result of $A_{\ell}$ with probability $\mu_{\ell}$, for all $\ell \in[k]$ and for any possible $\mu$. Now, given a fixed distribution $\mu,$ the expected approximation ratio is at least 
    \begin{align*}
        \alpha &\ge
         \;\sum_{\ell}\mu_{\ell} \frac{\sum_{e \in E} w_e\, a_{\ell}(b_e)}{\sum_{e \in E} w_e b_e} \\ 
         &=
         \;\sum_{\ell}\mu_{\ell} \sum_{e \in E^+} \left( \frac{w_e b_e}{\sum_{e \in E} w_e b_e} \right) \frac{a_{\ell}(b_e)}{b_e} \\ 
         &\ge
         \;\min_{e\in E+} \,\sum_{\ell}\mu_{\ell} \frac{ a_{\ell}(b_e)}{ b_e} \\ 
         &\ge
         \;\min_{b \in B \setminus \ofc{0}} \sum_{\ell}\mu_{\ell} \frac{ a_{\ell}(b)}{ b}\,,
    \end{align*}
where $E^+$ is the set of edges with $b_e >0$.
The third line follows because 
\begin{align*}
    \ofc{\frac{w_e b_e}{\sum_{e \in E} w_e b_e}}_{e \in E}
\end{align*} 
is a distribution over $E$. Since the statement is true for all possible distributions $\mu$, we take the maximum over $\mu$ to finish the lemma.
\end{proof}

We apply \cref{lem:minimax_reduction} with the SDP upper bound in \cref{eqn:sdpupperbound}
and our two algorithms $\match$ and $\pmatch$ as described in \cref{alg:algo_qmc}.  For $\match$, we derived in \cref{sec:preliminaries/previous_algorithms} that the energy obtained by $\match$ on $G$ is $\frac{3\maxM+W}{2}$. Since $\maxM$ is a maximum-weight matching with respect to $G$, we may invoke \cref{eq:s_matching_relation} to express
\begin{align*}
    \maxM \geq \sum_{(i,j)\in E} w_{ij} \,\cdot  d \cdot \of{-\frac{1+3s_{ij}}{2}}^+\,,
\end{align*}
implying that $\match$ achieves energy at least
\begin{align*}
    \frac{3\maxM+W}{2} \geq \sum_{(i,j)\in E} \frac{w_{ij}}{2} \of{1+3 d \of{-\frac{1+3s_{ij}}{2}}^+}\,.
\end{align*}
We have lower-bounded the energy that each algorithm earns on an edge $(i,j)$ as a function of $s_{ij}$. Note that the SDP upper bound is trivial (i.e. at most $0$) when $s_{ij} \ge \frac{1}{3}$, so we may search over $s \in [-1, 1/3)$. Altogether, \Cref{lem:minimax_reduction} implies that \Cref{alg:algo_qmc} obtains approximation ratio at least

\begin{align}
    \alpha &\ge \max_{\mu \in \ofb{0,1},\, \theta \in \ofb{0,\pi/2}}\of{ 
    \min_{s \in [-1,1/3)} \;\of{ \mu\, E_{PM}(\theta) + (1-\mu) E_M}}, 
    \label{eq:qmc_minimax_apx_rato}\\
    E_{PM}(\theta) &\defeq \frac{1-F(s) \cos^2\theta  
    + d \sin{\theta}\of{1-F(s)\of{1+\sin{\theta}}} \cdot  \big(-\frac{1+3s}{2}\big)^+}{1-3s}\,, \label{eq:e1_qmc_ratio}\\
    E_M(s) &\defeq
    \frac{1+3d\left(-\frac{1+3s}{2}\right)^+}{1-3s} \label{eq:e2_qmc_ratio}.
\end{align}
In the above, $\mu$ is demoted to a scalar (which fully describes a probability distribution over two variables). We do not need to consider the case $s=1/3$ because at this point the SDP upper bound is $0$ and the lower bound is always nonnegative. \cref{eq:e1_qmc_ratio,eq:e2_qmc_ratio} are the $\pmatch$ and $\match$ energies divided by the SDP energy, respectively.

One may then search over the three free parameters $\mu$, $\theta$, and $s$ using an enumeration of the feasible ranges. For a more efficient approach, it is also possible to solve a linear program where the number of variables and constraints is based on an enumeration of the feasible ranges for $s$ and $\theta$. We obtain $\alpha > 0.611$ for $d=14/15$ at $\theta=0.2847$ and $\mu = 0.861$. This would require \cref{alg:algo_qmc} to be run with $k=13$. If the factor $d$ in \cref{lem:strengthened_monogamy} is improved to $1$ (as described for the augmented SDPs in \cref{rem:d_equals_1}), we would obtain $\alpha > 0.614$ at $\theta=0.2845$ and $\mu = 0.881$. This can be done with $k=2$, augmented by a separation oracle as shown by \cite{apte2025b}.

\section*{Acknowledgments}
This work is supported by a collaboration between the US DOE and other Agencies.
The work of A.A is supported by the Data Science Institute at the University of Chicago. 
E.L is supported by a KIAS Individual Grant CG093801 at Korea Institute for Advanced Study. 
K.M acknowledges support from AFOSR
(FA9550-21-1-0008). K.M and J.S. acknowledge that this material is based upon work supported by the National Science Foundation Graduate Research Fellowship under Grant No.\ 2140001. 
O.P. acknowledges that this material is based upon work supported by the U.S. Department of Energy, Office of Science, Accelerated Research in Quantum Computing, Fundamental Algorithmic Research toward Quantum Utility (FAR-Qu). 
J.S acknowledges that this work is funded in part by the STAQ project under award NSF Phy-232580; in part by the US Department of Energy Office of Advanced Scientific Computing Research, Accelerated Research for Quantum Computing Program.

This article has been authored by an employee of National Technology \& Engineering Solutions of Sandia, LLC under Contract No.\ DE-NA0003525 with the U.S. Department of Energy (DOE). The employee owns all right, title and interest in and to the article and is solely responsible for its contents. The United States Government retains and the publisher, by accepting the article for publication, acknowledges that the United States Government retains a non-exclusive, paid-up, irrevocable, world-wide license to publish or reproduce the published form of this article or allow others to do so, for United States Government purposes. The DOE will provide public access to these results of federally sponsored research in accordance with the DOE Public Access Plan \url{https://www.energy.gov/downloads/doe-public-access-plan}.

\printbibliography
\appendix

\clearpage
\newpage

\section{\texorpdfstring{Proof of \cref{lem:qmc_energy_by_edge}}{Proof of QMC energy by edge}}\label{apx:proom_om_lemma_qmc_energy_by_edge}

We aim to find a 2-qubit quantum state $\rho_{ij}$, which we will refer to as $\rho$, whose rescaled single-qubit marginals align with the unit vectors $\bloch{i}$ and $\bloch{j}$. To this end, we use the \emph{Bloch matrix representation} of 2-qubit states from \cite{gamel2016}. Namely, we write 
\begin{align*}
    \rho = \frac{1}{4} \sum_{\mu,\nu \in \ofb{4}}r_{\mu\nu}D_{\mu\nu}\,, 
\end{align*}
where the scalars $r_{\mu\nu}$ constitute the $16$ entries of a \emph{Bloch matrix} $r$ and $D_{\mu\nu} \defeq \vec{\sigma}^{\prime} \otimes \vec{\sigma}^{\prime}$ and $\vec{\sigma}^{\prime} \defeq \of{I,X,Y,Z}$. The conditions that $\rho_{ij}$ is Hermitian and unit trace imposes that $r_{ij} \in \mathbb{R}$ and $r_{00}=1$, respectively. The positive semidefinite constraint is nontrivial and is fully characterized by \cite{gamel2016}. We denote the set of valid Bloch matrices by $\mathcal{B}$. It is useful to write out $r$ explicitly for $\rho$
\begin{align*}
    r =
\begin{bmatrix}
1      & r_{01} & r_{02} & r_{03} \\
r_{10} & r_{11} & r_{12} & r_{13} \\
r_{20} & r_{21} & r_{22} & r_{23} \\
r_{30} & r_{31} & r_{32} & r_{33}
\end{bmatrix}
\defeq
\begin{bmatrix}
1 & \widetilde{v}_j^T \\
\widetilde{v}_i & R
\end{bmatrix}\,,
\end{align*}
where in particular, $v_i$ and $v_j$ are the specified single-qubit marginal Bloch vectors of qubits $i$ and $j$, respectively, and $R$ is the $\emph{correlation matrix}$ encoding the expectation values of $2$-qubit Pauli matrices $\braket{XX}, \braket{XY}, \allowbreak \ldots, \braket{ZZ}$. 

Under this notation, the energy of $\rho$ with respect to a QMC term is
\begin{align*}
    \tr\ofb{h^{QMC} \rho} = \frac{1-\tr\ofb{R}}{2}.
\end{align*}
Maximizing this energy given $v_i$, $v_j$, and $\theta$ corresponds to the following optimization problem:
\begin{align}\label{eq:two_qubit_state_optimization_pre_svd}
    \min_{r \in \mathcal{B}} \quad\quad &\tr\ofb{R} \\
    \text{subject to} \quad &\widetilde{v}_i = \cos{\theta}\, v_i \\
     &\widetilde{v}_j = \cos{\theta}\, v_j 
\end{align}
This can be solved exactly using the second level of the quantum moment-SOS hierarchy, which is exact for two qubit systems. The constraints impose linear constraints on 1-local moments, and the objective is a sum of 2-local moments.

For our purposes it will suffice to produce a feasible solution meeting the conditions stated in \cref{lem:qmc_energy_by_edge}. While an optimal solution to the above problem may be a mixed state, we will focus our attention on pure states for simplicity. An extensive numerical assessment suggests that the worst-case approximation ratio presented in this work cannot be improved by considering states beyond the one we construct. We use a characterization of pure two qubit states from \cite{gamel2016}, based on a singular value decomposition of $R$. The Bloch matrix of any such state can be expressed as: 
\begin{align*}
    R = U \Sigma V^T\,, \quad \widetilde{v}_i = Ug\,, \quad \widetilde{v}_j = Vh\,,
\end{align*}
where $U$ and $V$ are real-valued orthogonal $3 \times 3$ rotation matrices, and
\begin{align}\label{eq:pure_state_characterization}
    g = h = \of{\cos{\theta}, 0, 0}^T\,, \\
    \Sigma = \mathrm{diag}\of{1, \sin\theta, \sin\theta}\,, \\
    \label{eq:negative-det}
    \mathrm{det}\of{U}\mathrm{det}\of{V} = -1\,,
\end{align}
where $\theta$ above will coincide with our parameter $\theta$. We begin constructing $\rho$ with the desired properties by first considering the state $\rho'$ obtained by taking $U' = I$, and
\begin{align*}
    V' =
\begin{bmatrix}
t & -\sqrt{1-t^2} & 0 \\
\sqrt{1-t^2} & t & 0 \\
0 & 0 & -1
\end{bmatrix}\,,
\end{align*}
where $t = v_i \cdot v_j$. We have $\mathrm{det}\of{U'} = 1$ and  $\mathrm{det}\of{V'} = -1$, satisfying \cref{eq:negative-det}. We also have $\widetilde{v}'_i \cdot \widetilde{v}'_j = \cos^2\!{\theta}\,t$ and that both $\widetilde{v}'_i$ and $\widetilde{v}'_j$ have magnitude $\cos\theta$. We can thus pick an orthogonal matrix $W$ such that $\cos{\theta}\, v_i = W\widetilde{v}_i'$ and $\cos{\theta}\, v_j = W\widetilde{v}_j'$.

We get $\rho$ by taking $U = W$ and $V = WV'$. In this case $R = W\Sigma(V')^TW^T$ and $\tr\ofb{R} = \tr\ofb{\Sigma (V')^T} = t - \sin\theta(1-t)$. This gives the desired energy of
\begin{align*}
\frac{1-\tr\ofb{R}}{2} = \frac{(1+\sin\theta)(1-t)}{2}\,.
\end{align*}

\clearpage
\newpage
\section{\texorpdfstring{Proof of \Cref{lemma:epr_numerical_inequality}}{Proof of EPR approximation ratio}}
\label{apx:proom_om_lemma_epr_numerical}
In the following proof, we use $c$ to denote the golden ratio $\frac{1+\sqrt{5}}{2}$ instead of  $\varphi$ as in the main text. This change is solely for ease of readability. 

\begin{proof}
In terms of $c$, the statement we would like to prove is
\begin{align*}
   \min_{x \in [0,1]} \ofb{
    \frac{1}{2(1+x)} \left(
    1 + c^{-(1-x)} + 2\sqrt{1 - c^{-x}}
    \cdot c^{-\frac{1}{2}(1-x)}
    \right)
    } =\frac{c}{2}.
\end{align*}
The golden ratio $c$ obeys the identities $c^2 = c + 1$ and $c^3 = 2c + 1$. Note that the minimum value of \cref{eq:epr_min_problem} is at most $\frac{c}{2}$. This is because when $x = 0$, the left-hand side is equal to 
\begin{align*}
    \frac{1}{2} \left(1 + c^{-1} + 0\right) = \frac{1}{2c}(c + 1) = \frac{c^2}{2c} = \frac{c}{2}\,.
\end{align*}
Thus it suffices to show that for all $x \in [0,1]$,
\begin{align*}
1 + c^{-(1-x)} + 2\sqrt{1 - c^{-x}}
\cdot c^{-\frac{1}{2}(1-x)}
\stackrel{?}\ge c(1+x)\,.
\end{align*}
We repeatedly transform this inequality to find a simpler, equivalent inequality:
\begin{align*}
    &\Leftrightarrow  & 
c + c^{x}    + 2\sqrt{1 - c^{-x}}
\cdot c^{\frac{1}{2}(x+1)}
&\quad \stackrel{?}\ge\quad  c^2(1+x) \\
    &\Leftrightarrow  &
c + (c^{x}-1) + 2\sqrt{c} \cdot \sqrt{c^x - 1}
 &\quad \stackrel{?}\ge\quad  (c+1)(1+x)-1  \\
    &\Leftrightarrow  &
(\sqrt{c} + \sqrt{c^x - 1})^2
&\quad \stackrel{?}\ge\quad  cx + x + c  \\
    &\Leftrightarrow  &
\sqrt{c} + \sqrt{c^x - 1}
&\quad \stackrel{?}\ge\quad \sqrt{c^2x + c} \\
    &\Leftrightarrow  &
\sqrt{c^x - 1}
&\quad \stackrel{?}\ge\quad \sqrt{c} \cdot (\sqrt{cx + 1} - 1) \\
    &\Leftrightarrow  &
c^x - 1
&\quad \stackrel{?}\ge\quad c \cdot (cx + 1 + 1 - 2\sqrt{cx + 1}) \\
 &\Leftrightarrow  &
2c\sqrt{cx+1}
&\quad \stackrel{?}\ge\quad c^2 x + 2c + 1 -c^x \\
 &\Leftrightarrow  &
\sqrt{4cx+4}
&\quad \stackrel{?}\ge\quad c x + c^2 -c^{(x-1)} \\
 &\Leftrightarrow  &
4cx+4
&\quad \stackrel{?}\ge\quad
(cx+c^2)^2 +
c^{(2x-2)} 
-2 c^{(x-1)}(cx + c^2) \\
 &\Leftrightarrow  &
0
&\quad \stackrel{?}\ge\quad
c^2(x+c)^2
-4cx-4
+
c^{(2x-2)}
-2c^x(x+c)\,.
\end{align*}
Let $f(x) \defeq c^2(x+c)^2
-4cx-4
+
c^{(2x-2)}
-2c^x(x+c)$. We will complete the proof by showing that $f(x) \le 0$ for all $x \in [0,1]$. We do this by proving that $f$ is convex in the region $x \in [0,1]$, so its maximum value is attained at either $x=0$ or $x=1$. We then verify the inequality at each boundary point. Consider the second derivative of $f$ with respect to $x$:
\begin{align*}
    f''(x) = 2c^2 + (2 \ln c)^2 c^{(2x - 2)} - 2c^x\ln c \cdot \left((x+c)\ln c + 2\right)\,.
\end{align*}
We substitute in the variable $z \defeq c^x$ for $z \in [1,c]$. Then $g(z) \defeq f''(\log_c x)$ is equal to
\begin{align*}
    g(z) =  2c^2 
    +
    z^2 \cdot (4 c^{-2} \ln^2 c)
    - z \cdot 2 \ln c \cdot ((c + \log_c z)\ln c + 2)\,.
\end{align*}
To show $f$ is convex, we want to show $f''(x) \ge 0$ for all $x \in [0,1]$ (or equivalently, $g(z) \ge 0$ for all $z \in [1,c]$). We do so by showing that $g$ is decreasing on the interval $[1,c]$, and $g(c) \ge 0$. Consider the first derivative of $g$ with respect to $z$:
\begin{align*}
    g'(z) &= z \cdot (8c^{-2} \ln^2 c) - 4 \ln c - 2 c\ln^2 c - 2 \ln^2 c \cdot \left(\log_c z + \frac{1}{\ln c} \right)\,.
\end{align*}
Note that $g'$ is itself convex. Using the higher-order derivatives of $g$, 
\begin{align*}
 g''(z) &= 8c^{-2} \ln^2 c - \frac{1}{z \ln c} \cdot 2 \ln^2 c = 8c^{-2} \ln^2 c - \frac{2 \ln c}{z}\,,\\
 g'''(z) &= \frac{2 \ln c}{z^2}\,,
\end{align*}
observe that $(g')'' = g''' \ge 0$ for all $z$. So $g'$ achieves its maximum value at a boundary point (either $z = 1$ or $z = c$). When $z = 1$, we have
\begin{align*}
    g'(1) &=  8c^{-2} \ln^2 c - 4 \ln c - 2 c\ln^2 c - 2 \ln^2 c \cdot \left(\frac{1}{\ln c} \right)
    \\
    &= (\ln c) \cdot (8c^{-2} \ln c - 6 - 2c \ln c)
    \\
    &\le (\ln c) \cdot (-6 + \frac{8 \ln c}{c+1})   
    \\
    &\le (\ln c) \cdot (-6 + \frac{8 \ln e}{2})
    \\
    &< 0\,.
\end{align*}
When $z = c$, we have
\begin{align*}
    g'(c) &=  8c^{-1} \ln^2 c - 4 \ln c - 2 c\ln^2 c - 2 \ln^2 c \cdot \left(1 + \frac{1}{\ln c} \right)
    \\
    &= (\ln c) \cdot (8c^{-1} \ln c - 6 - 2 (c+1) \ln c)
    \\
    &\le (\ln c) \cdot (-6 + \frac{8 \ln c}{c})
    \\
    &\le (\ln c) \cdot (-6 + \frac{8 \ln e}{1.6})
    \\
    &< 0\,.
\end{align*}
Since both of these values are negative, $g'$ is negative on the interval $[1,c]$. So $g$ is decreasing on the interval $[1,c]$.
To finish the claim that $g \ge 0$ on the interval $[1,c]$, we check that $g(c) \ge 0$:
\begin{align*}
    g(c) &= 2c^2 + 4 \ln^2 c - 2c \ln c \cdot ((c+1)\ln c + 2)
    \\
    &= 2c^2 - 4c \ln c + 4 \ln^2 c - 2c^3 \ln^2 c 
    \\
    &= 2(c - \ln c)^2 + 2 \ln^2 c - 2(2c+1) \ln^2 c 
    \\
    &= 2(c - \ln c)^2 - 4c \ln^2 c \,.
\end{align*}
Note that $1.6 < c < 2$ and $2 \ln c = \ln c^2 = \ln (c + 1) < \ln 2.62 < 1$. Then we can lower-bound $g(c)$ as 
\begin{align*}
    g(c) \ge 2(1.6 - 0.5)^2 - 4 \cdot 2 \cdot 0.5^2 = 2.42 - 2 > 0\,.
\end{align*}
Finally, we verify that the boundary points of $f$ are at most $0$. When $x = 0$, we have
\begin{align*}
    f(0) &= c^4 - 4 + c^{-2} - 2c 
    \\
    &= c^{-2}(c^6 - 2c^3 - 4c^2 + 1)
    \\
    &= c^{-2}((c^3 - 1)^2 - 4c^2)
    \\
    &= c^{-2}((2c)^2 - 4c^2)
    \\
    &= 0\,.
\end{align*}
When $x = 1$, we have
\begin{align*}
        f(1) &= c^2(c+1)^2 - 4c - 4 + 1 - 2c(c+1) 
        \\
        &= c^6 - 2c^3 - 4c - 3
        \\
        &= (c^3 - 1)^2 - 4(c+1)
        \\
        &= (2c)^2 - 4c^2
    \\
    &= 0\,.
\end{align*}
Since the maximum value of a convex function is attained on its boundary, we have that for all $x \in [0,1]$, $f(x) \le \max_{x \in [0,1]} f(x) = \max(f(0), f(1)) = 0$.
\end{proof}

\end{document}